\documentclass[12pt]{article}
\usepackage{amsmath,leftidx}
\usepackage{amssymb}
\usepackage{amsthm}
\usepackage{fullpage}
\usepackage[utf8]{inputenc}
\usepackage{mathtools}
\usepackage{url}
\usepackage[usenames,svgnames]{xcolor}
\usepackage{cite}
\usepackage{enumitem}
\usepackage{datetime}
\usepackage[titletoc,title]{appendix}
\usepackage{comment}
\usepackage{mathdots}
\usepackage{graphicx}

\usepackage[margin=1.3in]{geometry}

\usepackage{ tikz, pgfplots, todonotes,color}
\usetikzlibrary{decorations.markings,arrows,}

\usepackage{mathtools,cancel}
\usepackage{url}

\usepackage{todonotes}
\usepackage{comment}
\usepackage{mathdots}
\usepackage{graphicx}
\usepackage{float}
\usepackage{enumitem}

\usepackage{stackengine}
\usepackage{scalerel}

\usepackage[pagebackref=false]{hyperref} 
\usepackage[all]{hypcap} 

\def \z{\zeta}
\theoremstyle{plain}
\newtheorem{theorem}{Theorem}[section]

\newtheorem{definition-theorem}[theorem]{Definition-Theorem}
\newtheorem{definition-proposition}[theorem]{Definition-Proposition}

\newtheorem{example}{Example}[section]
\newtheorem{examples}{Example}[subsection]

\newtheorem{remark}{Remark}[section]

\theoremstyle{definition}
\newtheorem{definition}{Definition}[section]

\numberwithin{equation}{section} 

\DeclareMathOperator{\tr}{tr}

\DeclareMathOperator{\Span}{span}

\DeclareMathOperator{\cyc}{cyc}

\DeclareMathOperator{\aut}{aut}

 \def\script{\scriptstyle}

\def\ra{{\rightarrow}}

\def\tr{\mathrm {tr}}
\def\det{\mathrm {det}}

\def\Nor{\mathrm {Nor}}

\def\diag{\mathrm {diag}}

\def\res{\mathop{\mathrm {res}}\limits}

\def\be{\begin{equation}}
\def\ee{\end{equation}}

\def\bea{\begin{eqnarray}}
\def\eea{\end{eqnarray}}

\def\bt{\begin{theorem}}
\def\et{\end{theorem}}

\def\bex{\begin{example}\small \rm}
\def\eex{\end{example}}

\def\bexs{\begin{examples}\small \rm}
\def\eexs{\end{examples}}

\def\ra{\rightarrow}

\def\br{\begin{remark}\small \rm}
\def\er{\end{remark}}

\def\res{\mathop{\mathrm{res}}\limits}

\def\&{&{\hskip -20pt}}

\def\CC{\mathcal{C}}
\def\DD{\mathcal{D}}

\def\LL{\mathcal{L}}
\def\MM{\mathcal{M}}

\def\WW{\mathcal{W}}

\def\Cb{\mathbf{C}}

\def\Ib{\mathbf{I}}

\def\Nb{\mathbf{N}}

\def\Nb{\mathbf{N}}
\def\Pb{\mathbf{P}}

\def\Zb{\mathbf{Z}}

\def \d {\mathrm d}

\begin{document}
\baselineskip 16pt
\medskip
\begin{center}
\begin{Large}\fontfamily{cmss}
\fontsize{17pt}{27pt}
\selectfont
	\textbf{Rationally weighted Hurwitz numbers, \\ Meijer $G$-functions  and matrix integrals}
	\end{Large}
\\
\bigskip \bigskip
\begin{large}  M. Bertola$^{1, 2, 3}$\footnote{e-mail: Marco.Bertola@concordia.ca, Marco.Bertola@sissa.it} and J. Harnad$^{1, 2}$\footnote {e-mail: harnad@crm.umontreal.ca  }
 \end{large}\\
\bigskip
\begin{small}
$^{1}${\em Department of Mathematics and Statistics, Concordia University\\ 1455 de Maisonneuve Blvd.~W.~Montreal, QC H3G 1M8  Canada}\\
\smallskip
$^{2}${\em Centre de recherches math\'ematiques, Universit\'e de Montr\'eal, \\C.~P.~6128, succ. centre ville, Montr\'eal, QC H3C 3J7  Canada}\\
 \smallskip
$^{3}${\em SISSA/ISAS, via Bonomea 265, Trieste, Italy }
\smallskip
\end{small}
\end{center}

\begin{small}\begin{center}
\end{center}\end{small}
\medskip
\begin{abstract}
The quantum spectral curve equation associated to KP $\tau$-functions of hypergeometric type serving as generating functions for rationally weighted Hurwitz numbers is solved by generalized hypergeometric series.  The basis elements spanning the  corresponding  Sato Grassmannian element are shown to be Meijer $G$-functions, or their asymptotic series. Using their Mellin integral representation  the $\tau$-function, evaluated at the  trace invariants of an externally coupled matrix, is expressed as  a matrix integral.  \end{abstract}

\section{Introduction}

Several instances of  $\tau$-functions of hypergeometric type are known to be generating functions for
weighted Hurwitz numbers \cite{GH1, GH2, HO, H1}, which enumerate branched covers of the Riemann sphere with specified 
branching profiles. These are distinguished by the choice of a weight generating function $G(z)$ that selects the 
weighting attributed to various possible branching configurations. 

In particular, such $\tau$-functions serve as generating functions for  {\em simple} Hurwitz numbers \cite{Ok, OP} 
(both single and double),  {\em weakly monotonic} (or, equivalently, {\em signed})
Hurwitz numbers (\cite{GGN1, HO, GH2}),   {\em strongly monotonic} Hurwitz numbers \cite{GH1, GH2}, 
weighted Hurwitz numbers for Belyi curves, or {\em dessins d'enfants} \cite{GJ, AC1, KZ, Z},
polynomially weighted Hurwitz numbers \cite{AC1, AC2, AMMN, HO, ACEH1, ACEH2, ACEH3,  AC3}; 
{\em quantum} Hurwitz numbers \cite{H1, GH2, H2} and  {\em multispecies} Hurwitz numbers \cite{H3}. 

Of these,  simple Hurwitz numbers, both weakly  and strongly monotonic Hurwitz numbers and
polynomially weighted Hurwitz numbers are all known to admit  matrix integral 
representations  \cite{Or, BM, BEMS, AMMN, GGN1, AC1, AC2, AC3}.
In this work, we derive a new matrix integral representation, of a different type, for hypergeometric KP $\tau$-functions
 generating weighted Hurwitz numbers with rational weight generating function,
which include several of the above as special cases.

In Section \ref{gen_fns_weighted_Hurwitz}, the notion of weighted Hurwitz numbers \cite{GH1, HO, H1, GH2},
  is recalled and the  KP $\tau$-function serving as generating function for these is defined, focussing on the case  of rational 
 weight generating functions
 \be
G_{{\bf c}, {\bf d}}(z) :={ \prod_{i=1}^L (1+c_i z) \over \prod_{j=1}^M(1-d_j z)}.
\label{G_ratl_c_d_intro}
\ee
We define a natural basis $\{\phi_k^{({\bf c},{\bf d},\beta)} (x) \}_{k\in \Nb^+}$ 
 for the corresponding element of the infinite Sato Grassmannian, 
consisting of Laurent series  obtained by dividing generalized hypergeometric series of the type
\be
\leftidx{_L}{F_M}\left( {1-k+{1 \over \beta c_1}, \cdots, 1-k+{1\over \beta c_L} \atop
1-k-{1 \over \beta d_1}, \cdots,1-k- {1 \over \beta d_M}}  \bigg{|}  \kappa_{{\bf c}, {\bf d}} x \right), \quad k \in \Zb
\ee
by the monomials $x^{k-1}$. These are either convergent or formal, depending on whether $L \le M+1$ or $L>M+1$.
 They satisfy the sequence of eigenvalue equations with nonnegative integer eigenvalues,
 \be
\LL_{{\bf c} ,{\bf d}, \beta} \phi_k^{({\bf c},{\bf d},\beta)}  := (x G_{{\bf c}, {\bf d}}(\beta\DD) - \DD)\phi_k^{({\bf c},{\bf d},\beta)}   = (k-1)\phi_k^{({\bf c},{\bf d},\beta)}, \quad k \in \Zb,
\label{phi_k_eq_intro}
\ee
 where $\DD=x {d\over dx}$ is the Euler operator and $\LL_{{\bf c} ,{\bf d}, \beta}$  the quantum spectral curve operator. The first of  these  is the quantum curve equation appearing in the topological recursion approach \cite{ACEH1, ACEH2, ACEH3}.  
  
Theorem \ref{phi_k_G_function},  Section \ref{quantum_curve_solutions}, relates
the basis elements $\{\phi_k^{({\bf c},{\bf d},\beta)} (x) \}_{k\in \Nb^+}$
 to Barnes-Mellin type integrals 
 \be
\tilde{\phi}_k^{({\bf c},{\bf d},\beta)} (x) \sim  \int_{\CC_k}  \frac {\Gamma(1-k-s)  \prod_{\ell=1}^L \Gamma \left( s + \frac 1 {\beta c_\ell} \right) 
\left(-\kappa_{{\bf c}, {\bf d}} x\right)^s}{\prod_{m=1}^M \Gamma\left( s - \frac 1{\beta d_m}  \right)}  ds 
\label{phi_tilde_k_integ_intro}
\ee
 defining Meijer $G$-functions, either as convergent sums over the residues at the poles at $s=1-k, 2 -k, \cdots$ (when $L\le M+1$) or as asymptotic series (when $L> M+1$).   In Section \ref{tau_trace_invars} a finite determinantal formula  (\ref{tau_phi_i_det}) is derived for the $\tau$-function  with
KP flow parameters evaluated at the trace invariants of an $n$-dimensional matrix $X$.  
This leads to Theorem \ref{hypergeom_tau_matrix_integral},  expressing the $\tau$-function as an externally coupled matrix 
integral of generalized Br\'ezin-Hikami  \cite{BH} type.

\section{Generating functions for weighted Hurwitz numbers}
\label{gen_fns_weighted_Hurwitz}

\subsection{Pure and weighted Hurwitz numbers}

We recall the definition of {\em pure} Hurwitz numbers \cite{Frob1, Frob2, Sch, Hu1, Hu2, LZ}.
\begin{definition}[Combinatorial]
For a  set of $k$ partitions $\{\mu^{(i)}\}_{i=1,\dots, k}$ of $N\in \Nb^+$, the {\em pure} Hurwitz number 
$H(\mu^{(1)}, \dots, \mu^{(k)})$  is ${1\over N!}$ times the number of distinct ways that the identity element $\Ib_N \in S_N$ 
in the symmetric group in $N$ elements can be expressed as a product
\be
\Ib_N = h_1, \cdots h_k
\ee
of $k$ elements $\{h_i\in S_N\}_{i=1, \dots, k}$,  such that for each $i$, $h_i$ belongs
to the conjugacy class $\cyc(\mu^{(i)})$ whose cycle lengths are equal to the parts of $\mu^{(i)}$:
\be
h_i \in \cyc(\mu^{(i)}), \quad i =1, \dots, k.
\ee
\end{definition}

An equivalent definition consists of enumeration of branched coverings of the Riemann sphere.
\begin{definition}[Geometric]
For a set of partitions $\{\mu^{(i)} \}_{i=1,\dots, k}$ of weight $|\mu^{(i)}|=N$, 
the pure Hurwitz number  $H(\mu^{(1)}, \dots, \mu^{(k)})$ is defined geometrically  \cite{Hu1, Hu2}  as the number
of inequivalent $N$-fold branched coverings  $\CC \ra \Pb^1$  of the Riemann sphere with $k$ branch points $(Q^{(1)}, \dots, Q^{(k)})$, 
whose ramification profiles are given by the partitions $\{\mu^{(1)}, \dots, \mu^{(k)}\}$, 
normalized by the inverse  $1/|\aut (\CC)|$ of the order of the automorphism group of the covering. 
\end{definition}
The equivalence of these two definitions follows \cite{LZ} from  the monodromy homomorphism 
\be
\MM:\pi_1 (\Pb^1/\{Q^{(1)}, \dots, Q^{(k)})\} \ra S_N
\ee
from the fundamental group of the Riemann sphere punctured at the branch points into $S_N$,
 obtained by lifting closed loops from the base to the covering.
 
To define {\em weighted} Hurwitz numbers \cite{GH1, GH2, HO, H1}, we introduce a weight generating function $G(z)$, either as 
an infinite product
\be
G(z) =\prod_{i=1}^\infty (1 + c_i z)
\label{G_z_prod}
\ee
or an infinite sum
\be
G(z) = 1 + \sum_{i=1}^\infty g_i z^i,
\label{G_z_taylor}
\ee
either formally, or under suitable convergence conditions imposed upon the parameters.
Alternatively, it may be chosen in the dual form
\be
\tilde{G}(z) = \prod_{i=1}^\infty (1-c_i z)^{-1},
\label{G_tilde_z_prod}
\ee
which may also be developed as an infinite sum, 
\be
\tilde{G}(z)=  1 + \sum_{i=1}^\infty \tilde{g}_i z^i.
\label{G_tilde_z_taylor}
\ee

The independent parameters determining the weighting may be viewed as either $\{g_i\}_{i\in \Nb^+}$,
$\{\tilde{g}_i\}_{i\in \Nb^+}$ or $\{c_i\}_{i\in \Nb^+}$.
They are related by the fact that (\ref{G_z_prod}) and (\ref{G_tilde_z_prod}) are generating functions for 
elementary and complete symmetric functions, respectively,
\be
g_i = e_i({\bf c}), \quad \tilde{g}_i = h_i({\bf c}),
\ee
in the parameters ${\bf c} =(c_1, c_2, \dots)$.

In much of the following, we consider combined rational weight generating functions 
\be
G_{{\bf c}, {\bf d}}(z) :={ \prod_{i=1}^L (1+c_i z) \over \prod_{j=1}^M(1-d_j z)},
\label{G_ratl_c_d}
\ee
whose Taylor series expansion coefficients are
\be
g_i({\bf c}, {\bf d}) = \sum_{j=0}^i e_j({\bf c}) h_{i-j} ({\bf d}),
\label{g_i_coeffs_ratl}
\ee
where ${\bf c} =(c_1, \dots, c_L)$, ${\bf d} =(d_1, \dots, d_M)$. 

\begin{definition}[Weighted Hurwitz numbers]
For the case of a weight generating function expressed in the form (\ref{G_z_prod}), choosing a nonnegative integer $d$ and a fixed partition $\mu$
of weight  $|\mu| =N$, the weighted (single) Hurwitz number $H^d_G(\mu)$ is defined as the weighted sum over all $k$-tuples $(\mu^{(1)}, \dots, \mu^{(k)})$
\be
H^d_G(\mu) := 
\sum_{k=1}^d \sum_{\mu^{(1)}, \dots \mu^{(k)}, \  |\mu^{(i)}| =N  \atop \sum_{i=1}^k \ell^*(\mu^{(i)}) =d} 
\WW_G(\mu^{(1)}, \dots, \mu^{(k)}) H(\mu^{(1)}, \dots, \mu^{k)}, \mu)
\label{H_d_G_def}
\ee
where
\be
\ell^*(\mu^{(i)}) := |\mu^{(i)}| - \ell(\mu^{(i)})
\ee
is the {\em colength} of the partition $\mu^{(i)}$, and the weight factor is defined to be
\bea
 \WW_G(\mu^{(1)}, \dots, \mu^{(k)}) &\&:=
 {1\over k!}
 \sum_{\sigma \in S_{k}} 
 \sum_{1 \leq b_1 < \cdots < b_{k } \leq M} 
  c_{b_{\sigma(1)}}^{\ell^*(\mu^{(1)})} \cdots c_{b_{\sigma(k)}}^{\ell^*(\mu^{(k)})} \cr
&\& = {|\aut(\lambda)|\over k!} m_\lambda ({\bf c}).
 \label{WG_def}
\eea
Here  $m_\lambda ({\bf c}) $ is the monomial symmetric function \cite{Mac} of the parameters ${\bf c}:= (c_1, c_2, \dots)$
\be
m_\lambda ({\bf c}) = {1\over |\aut(\lambda)|}\sum_{\sigma \in \mathfrak{S}_{k}} \sum_{1 \leq b_1 < \cdots < b_{k }}
 c_{b_{\sigma(1)}}^{\lambda_1} \cdots c_{b_{\sigma(k)}}^{\lambda_{k}} ,
  \label{m_lambda}
\ee
 indexed by the partition $\lambda$ of weight $|\lambda|=d$ and length  $\ell(\lambda) =k$, whose 
parts $\{\lambda_i\}$  are equal to the colengths $\{\ell^*(\mu^{(i)})\}$ (expressed in weakly decreasing order) 
\be
\{\lambda_i\}_{i=1, \dots k} \sim \{\ell^*(\mu^{(i)})\}_{i=1, \dots k},  \quad \lambda_1 \ge \cdots \ge \lambda_k >0
\label{lambda_colengths_mu}
\ee
and
 \be
 |\aut(\lambda)|:=\prod_{i\geq 1} m_i(\lambda)!
 \ee
where $m_i(\lambda)$ is the number of parts of $\lambda$ equal to~$i$.  A similar definition applies for weight generating functions
$\tilde{G}(z)$ of the dual form (\ref{G_tilde_z_prod}), with the monomial symmetric functions (\ref{m_lambda}) appearing in
(\ref{H_d_G_def}) replaced by the {\em forgotten} symmetric functions  \cite{Mac}
\be
f_\lambda({\bf c}) := {(-1)^{\ell^*(\lambda)}\over |\aut(\lambda)|}\sum_{\sigma \in \mathfrak{S}_{k}} \sum_{1 \leq b_1 <\le\cdots \le b_{k }}
 c_{b_{\sigma(1)}}^{\lambda_1} \cdots c_{b_{\sigma(k)}}^{\lambda_{k}} ,
  \label{f_lambda}
\ee

For the case of rational weight generating functions $G_{{\bf c}, {\bf d}}(z)$,
the weighted (single) Hurwitz numbers are
\bea
H^d_{G_{{\bf c}, {\bf d}}}(\mu)&\& := 
\sum_{1\le k, l \atop k+l \le d} \sum_{{\mu^{(1)}, \dots \mu^{(k)} ,\nu^{(1)}, \dots \nu^{(l)}, \atop    
\sum_{i=1}^k \ell^*(\mu^{(i)})+ \sum_{j=1}^l \ell^*(\nu^{(j)})  =d}\atop  |\mu^{(i)}|  =   |\nu^{(j)}|=N }
{\hskip -30 pt} \WW_{G_{{\bf c}, {\bf d}}}(\mu^{(1)}, \dots, \mu^{(k)}; \nu^{(1)}, \dots, \nu^{(l)})  \cr
&\&{\hskip 120 pt}  \times H(\mu^{(1)}, \dots, \mu^{(k)}, \nu^{(1)}, \dots, \nu^{(l)}, \mu),\cr
&\&
\label{H_d_G_rat_def}
\eea
where the weight factor is
\bea
&\& \WW_{G_{{\bf c}, {\bf d}}}(\mu^{(1)}, \dots, \mu^{(k)}; \nu^{(1)}, \dots, \nu^{(l)}) \cr
&\& :=
{(-1)^{\sum_{j=1}^l \ell^*(\nu^{(j)}) -l }\over k! l!}
 \sum_{\sigma \in S_k \atop \sigma' \in S_l}   \sum_{1 \leq a_1 < \cdots < a_{k } \leq M\atop 1 \leq b_1\cdots \leq b_k\leq L } 
  c_{a_{\sigma(1)}}^{\ell^*(\mu^{(1)})} \cdots c_{a_{\sigma(k)}}^{\ell^*(\mu^{(k)})}  
  d_{b_{\sigma'(1)}}^{\ell^*(\nu^{(1)})} \cdots d_{b_{\sigma'(l)}}^{\ell^*(\nu^{(l)})}.\cr
  &\&
 \label{WG_def}
\eea
\end{definition}

The sum 
\be
d := \sum_{i=1}^k \ell^*(\mu^{(i)}),   \text{ or } \  d := \sum_{i=1}^k \ell^*(\mu^{(i)}) +  \sum_{j=1}^l \ell^*(\nu^{(j)})
\ee
 of the colengths of the weighted partitions is related to the Euler characteristic $\chi$ or genus $g$
 of the covering curve by the Riemann-Hurwitz formula
\be
\chi = 2-2g = N +\ell(\mu) - d.
\label{riemann_hurwitz}
\ee


\subsection{Hypergeometric $\tau$-functions as generating functions for weighted Hurwitz numbers}

We recall the definition of  KP $\tau$-functions of hypergeometric type \cite{OrSc1, OrSc2}
that serve as generating functions for weighted Hurwitz numbers \cite{GH2, HO, H1}.
For a weight generating function $G(z)$ and nonzero parameter $\beta$,
we define two doubly infinite sequences $\{r_i^{(G, \beta)}, \rho_i\}_{i \in \Zb}$,
labeled by the integers
\bea
r^{(G, \beta)}_i &\&:= \beta G( i\beta), \quad i \in \Zb,  \quad \rho_0 =1,
\label{r_G_beta_i_def} \\
\rho_i &\& := \prod_{k=1}^i r^{(G, \beta)}_k,\quad  \rho_{-i}:=\prod_{k=0}^{i-1}( r^{(G, \beta)}_{-k})^{-1}, 
\quad i\in \Nb^+
\label{rho_i_def}
\eea
related by
\be
r_i^{(G, \beta)} = {\rho_i \over \rho_{i-1}},
\ee
where $\beta$ is viewed as a small parameter for which $G(i\beta)$ does not vanish for
any integer $i\in  \Zb$. 

\begin{remark}The particular cases of {\em simple}, {\em strongly monotonic,} {\em weakly monotonic} Hurwitz numbers
and {\em Belyi curves }correspond, respectively, to weight generating functions $G(z)=e^z$, ${1\over 1-z}$, $1+z$ and
$(1+ c_1 z)(1+c_2 z)$.
\end{remark}

For general rational weight generating functions (\ref{G_ratl_c_d}),
the parameters $\{\rho_i\}_{i\in \Zb}$ become
\bea
\rho_i({\bf c},{\bf d}) &\&:= \beta^i\prod_{k=1}^i {\prod_{l=1}^L (1+k \beta c_l)  \over \prod_{m=1}^M (1-k\beta d_m) }, \cr
 \rho_{-i}({\bf c},{\bf d})&\&:= \beta^{-i}\prod_{k=1}^{i-1} {\prod_{m=1}^M (1+k\beta d_m ) \over \prod_{l=1}^L(1- k \beta c_l ) },
\quad i\in \Nb. 
\label{rho_c_b_i_def}
\eea
For each partition $\lambda$ of $N$, we define the associated {\em content product} coefficient
\be
r^{(G, \beta)}_\lambda := \prod_{(i,j) \in \lambda} r^{(G, \beta)}_{j-i},
\ee
where the product is over the locations of all boxes in the Young diagram for $\lambda$.
The KP $\tau$-function of hypergeometric type associated to these parameters is 
 defined as the Schur function series \cite{GH2, HO}
\be
\tau^{(G, \beta)}({\bf t}) 
:=\sum_{N=0}^\infty \sum_{{\lambda}, \ |\lambda|=N}  
(h(\lambda))^{-1} r_\lambda^{(G, \beta)}  s_\lambda({\bf t}),
\label{tau_G_beta_schur_series}
\ee
where $h(\lambda)$ is the product of the hook lengths of the partition $\lambda$
 and ${\bf t}=(t_1, t_2 \dots)$ is the infinite sequence of KP flow parameters, which may be equated to the 
 sequence  of normalized power sums $(p_1, {1\over 2}p_2, \dots)$,
\be
t_i := {1\over i} \sum_{a} x_a^i = {1\over i} p_i
\ee
 in a finite or infinite set of auxiliary variables $(x_1, x_2, \dots )$. 
 
 Using the Schur character formula \cite{FH, Mac}
  \be
  s_\lambda = \sum_{\mu, \ |\mu|=|\lambda| = N} {\chi_\lambda(\mu) \over z_\mu} p_\mu,
  \label{schur_character_formula}
  \ee
  where $\chi_\lambda(\mu)$ is the irreducible character of the $S_N$ representation 
  determined by $\lambda$ evaluated on the conjugacy class $\cyc(\mu)$ consisting of elements with cycle lengths
  equal to the parts of $\mu$,
  \be
  z_{\mu} := \prod_{i=1}^{\ell(\mu)} m_i(\mu)! i^{m_i(\mu)}
  \label{mu_stabilizer_order}
  \ee 
  is the order of the stabilizer of the elements of this conjugacy class,  and
  \be
  p_\mu := \prod_{i=1}^{\ell(\mu)} p_{\mu_i}
  \label{power_sum_mu}
  \ee
  is the power sum symmetric function corresponding to partition $\mu$, we may re-express the Schur function
   series (\ref{tau_G_beta_schur_series}) as an expansion  in the basis $\{p_\mu\}$.
  \begin{theorem}[\cite{GH1, GH2, HO, H1}]
  The $\tau$-function $\tau^{(G, \beta)}({\bf t}) $ may equivalently be expressed as
 \be
\tau^{(G, \beta)}({\bf t}) 
= \sum_{d=0}^\infty  \beta^d \sum_\mu H^d_G (\mu) p_\mu({\bf t}),
 \label{tau_G_beta_power_sum_series}
 \ee
and the particular case of  rational weight generating functions $G_{{\bf c}, {\bf d}}(z)$, as
 \be
\tau^{(G_{{\bf c}, {\bf d}}, \beta)}({\bf t}) 
= \sum_{d=0}^\infty  \beta^d \sum_\mu  H^d_{G_{{\bf c}, {\bf d}}}(\mu) p_\mu({\bf t}).
 \label{tau_G_c_d_beta_power_sum_series}
 \ee
    \end{theorem}
It is thus a generating function for the weighted Hurwitz numbers $H^d_G (\mu) $. (See also \cite{Or} for $\tau$-functions
of hypergeometric type and, in particular, the case of rational weight generating functions.)

\subsection{The (dual) Baker function and adapted basis}

By Sato's formula \cite{Sa, SS, SW}, the dual Baker-Akhiezer function $\Psi^*_{(G, \beta)}(z, {\bf t})$ corresponding
to the KP $\tau$-function (\ref{tau_G_beta_schur_series}) is
\be
\Psi^*_{(G, \beta)}(z, {\bf t}) = e^{-\sum_{i=1}^\infty t_iz^i} 
{\tau^{(G, \beta)}({\bf t} +[z^{-1}]) \over \tau^{(G, \beta)}({\bf t}) },
\ee
where
\be
[z^{-1}] := \left({1\over z}, {1\over 2 z^2}, \dots, {1\over n z^n}, \dots\right).
\ee
Evaluating at ${\bf t } ={\bf 0}$ and setting
\be
x := {1\over z},
\ee
we define
\be
\phi_1(x) := \Psi^*_{G, \beta}(1/z, {\bf 0}) = \tau^{(G, \beta)}([x]). 
\label{phi1_tau}
\ee

More generally \cite{ACEH1, ACEH2, ACEH3}, we introduce a doubly infinite sequence of  functions $\{\phi_k(x)\}_{k\in \Zb}$,
defined as  contour integrals around a small circle centred at  the origin (or formal residues)
\be
\phi_k(x) = {\beta\over 2\pi i  x^{k-1}} \oint_{|\zeta|=\epsilon} \rho_k(\zeta) e^{\beta^{-1} x\zeta } {d\zeta\over \zeta^k},
\label{phi_k_fourier_rep}
\ee
where $\rho_k(\zeta)$ is the semi-infinite Laurent  series in ${1\over \zeta}$ centred at the origin
\be
\rho_k(\zeta) := \sum_{i=-\infty}^{k-1} \rho_{-i-1} \zeta^i,
\label{rho_fourier}
\ee
with the $\rho_i$'s given by eqs.~(\ref{r_G_beta_i_def}), (\ref{rho_i_def}).
Then $\{\phi_k(1/z)\}_{k\in \Nb^+}$ forms a basis for the element $w^{(G,\beta)}$ of the infinite 
Sato-Segal-Wilson Grassmannian corresponding to the $\tau$-function $\tau^{(G,\beta)}({\bf t})$.

 The $\phi_k$'s may alternatively be expanded as formal (or convergent)  Laurent series by evaluating the integrals 
 as a sum of residues at the origin, 
\be
\phi_{k}(x)  = \beta x^{1-k}\sum_{j= 0}^{\infty} 
{ \rho_{j-k}\over j!} \left({x \over\beta}\right)^j.
\label{phi_k_series}
\ee
This has a nonvanishing radius of convergence in $x$ provided $\rho_j$ grows no faster than a power 
 \be
 |\rho_j | < K^j, \quad j \ra \infty.
 \ee

Following \cite{ACEH1, ACEH2, ACEH3}, we introduce the recursion operator
\be
R := \beta x G(\beta\DD),
\ee
where $\DD$ is the Euler operator
\be
\DD := x {\d \over dx},
\ee
and verify that the $\phi_k$'s satisfy the recursion relations
\be
R(\phi_k)= \phi_{k-1},  \quad  k \in \Zb.
\label{phi_k_rec}
\ee

\subsection{The quantum and classical spectral curve}

It is easily seen from the series expansions (\ref{phi_k_series}) that the  functions $\{\phi_k(x)\}_{k\in \Zb}$
satisfy the sequence of equations
\be
\beta (\DD +k -1) \phi_k= \phi_{k-1}, \quad k \in \Zb,
\label{D_phi_k}
\ee
which, combined with (\ref{phi_k_rec}), are equivalent to:
\be
\LL_{{\bf c} ,{\bf d}, \beta} \phi_k := (x G(\beta\DD) - \DD)\phi_k  = (k-1)\phi_k \quad k \in \Zb.
\label{phi_k_eq}
\ee
The case $k=1$
\be
\LL \phi_1(x)  =0
\label{quant_spec_curve}
\ee
is known in the {\em topological recursion} approach \cite{Eyn:book, EO-review, EO2, ACEH1, ACEH2, ACEH3}, 
as the {\em quantum spectral curve} equation.  Substituting  
\be
\beta {d\over dx} \ra y,
\ee
in (\ref{quant_spec_curve}) gives the classical spectral curve: 
\be
y = \beta G(xy).
\ee
(The residual factor $\beta$ is just an artifact, due to the fact that, for simplicity, we have
equated $\beta$ with another generating function parameter $\gamma$, which keeps track of the
weights of the partitions appearing in the expansion (\ref{tau_G_beta_schur_series});
 it can here simply be set equal to $1$.)
 
  When the weight generating function is a polynomial $G=G_{\bf c, {\bf 0}}$, 
eq.~(\ref{quant_spec_curve}) is an ordinary differential equation of degree $L$ with coefficients 
that are polynomials of degree $\leq L$.
More generally, if $G$ is a rational function  \hbox{$G= G_{{\bf c}, {\bf d}}$}, as in (\ref{G_ratl_c_d}), we denote the  $\phi_k$'s as
\be
\phi^{({\bf c},{\bf d},\beta)} _k(x) = \beta x^{1-k}\sum_{j= 0}^{\infty} 
{ \rho_{j-k} ({\bf c}, {\bf d} )\over j!} \left({x \over\beta}\right)^j.
\label{phi_k_series_ratl}
\ee
In particular, for $k=1$, this is the hypergeometric series
\bea
\phi^{({\bf c},{\bf d},\beta)}_1(x) &\&= 
 \sum_{j=0}^\infty {\prod_{l=1}^L ( {1\over \beta c_l})_j\over
\prod_{m=1}^M ( -{1\over \beta d_m})_j }{(\kappa_{{\bf c}, {\bf d}} x)^j\over j!} \cr
&\& =\leftidx{_L}{F_M}\left( {{1 \over \beta c_1}, \cdots, {1\over \beta c_L} \atop
-{1 \over \beta d_1}, \cdots,- {1 \over \beta d_M}}  \bigg{|}  \kappa_{{\bf c}, {\bf d}} x \right),
\label{phi_1_hypergeom}
\eea
where
\be
\kappa_{{\bf c}, {\bf d}}:=(-1)^{M} {\prod_{l=1}^L \beta c_l\over \prod_{m=1}^M \beta d_m}.
\ee
When $L\le M$, this converges for all $x$;  when $L=M+1$ it converges within a finite radius of $x=0$.
But if $L > M+1$, it is a divergent formal series.
In general, for any $k \in \Zb$,  $\phi_k$ is  proportional to $x^{1-k}$ times the hypergeometric series 
with parameters shifted by $k-1$,
\bea
\phi^{({\bf c},{\bf d},\beta)}_k(x) &\&= {\beta \rho_{-k}({\bf c}, {\bf d})\over (\kappa_{{\bf c}, {\bf d}}x)^{k-1}}
 \sum_{j=0}^\infty {\prod_{l=1}^L (1-k+ {1\over \beta c_l})_j\over
\prod_{m=1}^M (1-k -{1\over \beta d_m})_j }{(\kappa_{{\bf c}, {\bf d}} x)^j\over j!} \cr
&\& = {\beta \rho_{-k}({\bf c}, {\bf d})\over (\kappa_{{\bf c}, {\bf d}}x)^{k-1}}\leftidx{_L}{F_M}\left( {1-k+{1 \over \beta c_1}, \cdots, 1-k+{1\over \beta c_L} \atop
1-k-{1 \over \beta d_1}, \cdots,1-k- {1 \over \beta d_M}}  \bigg{|}  \kappa_{{\bf c}, {\bf d}} x \right),  \cr
&\&
\label{phi_k_hypergeom}
\eea
and hence is  also convergent, or not, depending on whether   $L\le M$, $L=M+1$ or $L > M+1$.

Defining an auxiliary function $\Phi^{({\bf c},{\bf d},\beta)}(x)$ such that
\be
\phi^{({\bf c},{\bf d},\beta)}_1 = \prod_{j=1}^M (1 - \beta d_j \DD) \Phi^{({\bf c},{\bf d},\beta)},
\label{Phi_def_ratl}
\ee
the quantum spectral curve equation may be expressed as
\be
x \prod_{i=1}^L (1 + \beta c_i \DD) \Phi^{({\bf c},{\bf d},\beta)}  - \DD \prod_{j=1}^M (1 - \beta d_j \DD) \Phi^{({\bf c},{\bf d},\beta)}  =0.
\label{quant_spec_curve_ratl}
\ee

\begin{remark}
Note that, although in the space of differentiable functions in a punctured neighbourhood of $x=0$, with  possible
branch point at $x=0$, the operator $\prod_{j=1}^L (1 - \beta d_j \DD)$ has a kernel of dimension $L$,  
if we restrict ourselves to formal (or convergent) 
Laurent series in $x$ and assume that  the $\beta d_j$'s are not integers, the kernel vanishes, 
making $\Phi^{({\bf c},{\bf d},\beta)} (x)$ unique.
 \end{remark}
 
\section{Solutions of the quantum spectral curve equation for rational $G(z)=G_{{\bf c}, {\bf d}}(z)$}
\label{quantum_curve_solutions}

Equation \eqref{quant_spec_curve_ratl} is of the type satisfied by Meijer $G$-functions \cite{BE, NIST}
\be
\left[ (-1)^{p - m - n} \;\z \prod_{j = 1}^p \left(\z\frac {\d}{\d \z} +1- a_j \right) - \prod_{j = 1}^q \left(\z\frac {\d}{\d \z} - b_j \right) \right] G^{m,n}(\z) = 0,
\label{MeijerGeq}
\ee
with the identifications 
\bea
\label{subs}
p&\&=L,\  q = M+1,\  \zeta=  (-1)^{L - m - n} \kappa_{{\bf c}, {\bf d}} x, \cr
a_j &\&= 1- \frac 1{\beta c_j},\quad b_1 = 0,\ \ b_{j+1} = \frac 1 {\beta d_j} \ \ (j\geq 1).
\eea
These all have Mellin-Barnes type integral representations of the form
\be
G_{p,q}^{m,n}\left( {a_1, \cdots, a_p \atop b_1, \cdots, b_q}  \bigg{|} \zeta \right)
 ={1\over 2 \pi i} \int_{\CC_{{\bf a} {\bf b}}}
 {\prod_{j=1}^m \Gamma(b_j-s) \prod_{k=1}^n \Gamma(1-a_k + s) \zeta^s  \over 
\prod_{j=m+1}^q \Gamma(1-b_j + s) \prod_{k=n+1}^p \Gamma(a_k- s) }ds, 
 \ee
 where $0\leq m\leq q, \ \ 0\leq n \leq p$ and the contour $\CC_{{\bf a} {\bf b}}$ lies between the poles at $\{s=b_j + i\}_{j=1, \dots, m, \, i\in \Nb^+}$
 and $\{s= a_k - i -1\}_{k=1, \dots, n, \, i\in \Nb^+}$.

By retracting the contour to the right-hand $s$--plane and evaluating the residues at the poles of the integrand, we obtain
a representation of the  Meijer-G function as a sum of powers of $\zeta$ within the set $\bigcup (b_j + \Nb)$. 
If $p\geq q$, this only yields a formal series, which may be interpreted as an asymptotic series
in suitable sectors of the complex $\zeta$-plane. However,  if  $p<q$, it gives  a convergent expansion. 

For  $k \in \Zb$, we now define the sequence of functions
\bea
\tilde{\phi}_k^{({\bf c},{\bf d},\beta)} (x)&\& :=  C_k^{({\bf c},{\bf d},\beta)} G_{L,M+1}^{1,L}\left( {1-\frac 1{\beta c_1}, \cdots, 1-\frac 1{\beta c_L} \atop 1-k,
1+\frac 1 {\beta d_1}, \cdots, 1+\frac 1{\beta d_M}}  \bigg{|}  -\kappa_{{\bf c}, {\bf d}} x \right)\nonumber \\
&\&={C_k^{({\bf c},{\bf d},\beta)}\over 2\pi i } \int_{\CC_k}  \frac {\Gamma(1-k-s)  \prod_{\ell=1}^L \Gamma \left( s + \frac 1 {\beta c_\ell} \right) 
\left(-\kappa_{{\bf c}, {\bf d}} x\right)^s}{\prod_{m=1}^M \Gamma\left( s - \frac 1{\beta d_m}  \right)}  ds. \cr
&\&
\label{phi_tilde_k_integ}
\eea
where
\be
C_k^{({\bf c},{\bf d},\beta)} :={\prod_{j=1}^M \Gamma(-{1\over \beta d_j}) \over (-\beta)^{k-1}\prod_{\ell=1}^L \Gamma({1\over \beta c_\ell})},
\label{C_dc_beta_norm}
\ee
with the contour $\CC_k$ chosen so that all the poles of the integrand at the integers $1-k, 2-k, \cdots$ are to its right
and the poles at $\{-i- {1\over \beta c_j}\}_{ j=1, \cdots L, \, i\in \Nb^+}$ are to its left. 

If $L\le M$ this gives a convergent value 
for all $ x\neq 0$ when the contour $\CC_k$ starts at $+\infty$ below the real axis, proceeds in a clockwise sense,
passes  around all the poles on the  real axis and ends at $+\infty$. If $L=M+1$ it also converges along this contour,
within the punctured disc 
\be
0 <|\kappa_{{\bf c}, {\bf d}} x| <1.
\ee
 If  $L\ge M+1$ the contour $\CC_k$
may  be chosen to start and end on the imaginary axis at $-i \infty$ and $+i\infty$.
The integral then converges provided that  
\be
|\arg(-\kappa_{{\bf c}, {\bf d}} x)| < \left({L-M+1\over 2}\right) \pi.
\label{sector}
\ee

We  then have the following theorem.
\begin{theorem}
\label{phi_k_G_function}
The  functions $\tilde{\phi}^{({\bf c},{\bf d},\beta)}_k(x)$ satisfy the sequence of equations (\ref{D_phi_k}),
the recursion relations (\ref{phi_k_rec}) and the eigenfunction equations (\ref{phi_k_eq}).

If $L \le M$,   the Laurent series (\ref{phi_k_series_ratl})  for $\phi_k^{({\bf c},{\bf d},\beta)} (x)$ is 
absolutely convergent for all $x\in \Cb$, $x \neq 0$, and
\be
\tilde{\phi}_k^{({\bf c},{\bf d},\beta)} (x) = \phi_k^{({\bf c},{\bf d},\beta)} (x).
\label{phi_k_tilde_phi_k}
\ee
 If $L=M+1$, it converges to $\tilde{\phi}_k^{({\bf c},{\bf d},\beta)} (x)$
  within the punctured disc
 \be
 0 < |\kappa_{{\bf c}, {\bf d}} x| <1.
 \ee
If $L > M+1$, eq.~(\ref{phi_k_series_ratl})  is the asymptotic series for $\tilde{\phi}_k^{({\bf c},{\bf d},\beta)} (x) $ 
within the angular sector 
\be
0<\arg(- \kappa_{{\bf c}, {\bf d}} x) < (L-M+1) {\pi \over 2}.
 \ee
\end{theorem}
\begin{proof}
\noindent {\bf Case  $L \le M +1$.}
We first prove the result  for  $k=1$ by evaluation of the residues of the integrand
at the poles $s = 0, 1, \dots$. These all come from the factor $\Gamma(-s)$ in the numerator of the integrand,
whose residues are
\be
\res_{s=j} \Gamma(-s)  = {(-1)^j\over j!}.
\ee
Substituting this, and evaluating all other factors in the integrand at the poles gives, by the Cauchy residue theorem,
\bea
\tilde{\phi}^{({\bf c},{\bf d},\beta)} _1(x) &\&= \sum_{j=0}^\infty  \prod_{l=1}^L {\Gamma(j+ {1\over \beta c_l} ) \over \Gamma({1\over \beta c_\ell})}
\prod_{m=1}^M {\Gamma(-{1\over \beta d_m}) \over \Gamma(j- {1\over \beta d_m} )} {(\kappa_{{\bf c}, {\bf d}} x)^j\over j!}\cr
&\& = \sum_{j=0}^\infty {\prod_{l=1}^L \left({1\over \beta c_l}\right)_j \over \prod_{m=1}^M \left(- {1\over \beta d_m}\right)_j}
{(\kappa_{{\bf c}, {\bf d}} x)^j \over j!} =  \leftidx{_L}{F_M}
\left( {{1 \over \beta c_1}, \cdots, {1\over \beta c_L} \atop
-{1 \over \beta d_1}, \cdots,- {1 \over \beta d_M}}  \bigg{|}  \kappa_{{\bf c}, {\bf d}} x \right). \cr
&\&
\label{sum_Gamma_residues_0}
\eea
The proof for general $k\in \Zb$ is the same, {\em mutatis mutandis}, with the contour $\CC_1$ replaced by $\CC_k$,  the 
normalization constant $C_1^{({\bf c},{\bf d},\beta)}$ by $C_k^{({\bf c},{\bf d},\beta)}$ and the term $\Gamma(-s)$
in the numerator of the integrand by $\Gamma(1-k-s)$. (Also see Remark \ref{slater_L_smaller} below.)

\noindent {\bf Case  $L >M +1$.}  
In this case $p=L>  q=M+1$, the  ODE \eqref{MeijerGeq} has an irregular singularity at $\z=0$ and the solutions manifest Stokes' phenomenon.  The function \eqref{phi_tilde_k_integ} is one such solution. (See Remark \ref{slater_L_larger} below.)
We now  show that the series \eqref{phi_k_series_ratl} is indeed its asymptotic series as $x\to 0$ in the sector \eqref{sector}.  

First note again that the integrand in \eqref{phi_tilde_k_integ} has poles at $s = -k+1, -k+2,\dots$ with residues that  give 
precisely the corresponding coefficient of $x^s$ in \eqref{phi_k_series_ratl}. What must be shown is that the remainder 
of $\tilde{\phi}_k^{({\bf c},{\bf d},\beta)} (x)$ minus a truncation of the series \eqref{phi_k_series_ratl} to order $\mathcal O(x^N)$ is of order $o(x^N)$ as $x\ra 0 $ in the indicated sector, since this is the definition of Poincar\'e\ asymptotics. 

To show this, consider  the integral from \eqref{phi_tilde_k_integ}:
\bea
\label{phikj}
F(\z) := \int_{\CC_k}  \frac {\Gamma(1-k-s)  \prod_{\ell=1}^L \Gamma \left( s + \frac 1 {\beta c_\ell} \right) 
\z^s}{\prod_{m=1}^M \Gamma\left( s - \frac 1{\beta d_m}  \right)}  ds ,
\eea
where we have set $\z=-\kappa_{{\bf c}, {\bf d}} x$ for brevity.
We use the Stirling approximation in the following form 
\bea
\label{estimate}
&\& \ln \Gamma(s) = \left(s-\frac 1 2\right )\ln \left(s-\frac 1 2\right) - \left(s-\frac 1 2\right) + \mathcal O(1)\\
&\& \ln\Gamma(1-s) = -  \left(s-\frac 1 2\right )\ln \left(s-\frac 1 2\right) + \left(s-\frac 1 2\right) -\ln \sin(\pi s)  + \mathcal O(1).
\eea
Denoting the integrand of  \eqref{phikj} $A_k(s) \z^{s}$, with $\z =- \kappa_{\bf c ,d } x$,  we can estimate its modulus as follows: 
\be
\label{lnA}
\ln (A_k(s) \z^{s}) =(L-M-1)  \big(s \ln s - s\big) - \ln \sin(\pi s) +  s\ln\z +   o(|s|).
\ee
We choose the contour of integration $\mathcal C_k$ in \eqref{phikj} to tend to infinity in
both positive and negative vertical directions and approach $\Re s=N + \frac 1 2 $, $N\in \mathbb R$ in both.
Concretely, we  may take a contour that coincides with the vertical line $\Re s=N+\frac 1 2 $ for $|\Im s|$  sufficiently large,
but again passes to the left of the poles at $1-k, 2-k, \cdots$ and to the right of those at $\{-j - {1\over \beta c_l}\}_{l=1, \cdots ,L, \, j\in \Nb}$. 
Along this contour the real part of \eqref{lnA} has the following asymptotic behaviour
\bea
\ln| A(s) \z^s| &\& = -(L-M-1)|\Im s| \frac \pi 2 - \pi |\Im s| +| \Im s| \arg \z +o(|s|) \cr
&\& = \frac \pi 2 \left(M - 1 -L\right) |\Im s| +| \Im s| \arg \z + o(|s|).
\eea
Thus,  since the modulus of the integrand tends to zero faster than any inverse power of $|s|$,  
the integrand is convergent  as long as $|\arg \z| < \frac {\pi(L-M+1)}2$.

Using the Cauchy residue theorem we can retract (see Fig. \ref{Ck}) the contour $\mathcal C_k$ to $ \Re s = N+\frac 1 2$. In doing so we pick up finitely many residues at $s = -k+1,-k+2,\dots, N-1,N$, which correspond to the truncation of the series \eqref{phi_k_series_ratl} 
up to the power $x^N$ included. It only remains to estimate the remaining integration $R(\z)$ (remainder term):
\be
R(\z) =   \int_{i\mathbb R+N+\frac 1 2} A_k(s) \z^s d s
\ee
The estimate  \eqref{estimate} shows that  the integral is well defined as long as $|\arg \z| < \frac {\pi(L-M+1)}2$.
Furthermore the whole integral is easily estimated to be of the order $|\z|^{N+\frac 1 2}$ because
\bea
\left|\int_{i\mathbb R+N+\frac 1 2} A_k(s) \z^s d s\right|\leq |z|^{N+ \frac 1 2} \int_{i\mathbb R+N+\frac 1 2} |A(s)| {\rm e}^{\arg \z \Im s} \d s
\eea
and the last integral can be uniformly estimated in any sector $|\arg \z| < \frac {\pi(L-M+1)}2 - \epsilon$.   This shows that $\phi_k$ has the  formal series \eqref{phi_k_series} as asymptotic series in the sense of Poincar\'e.  

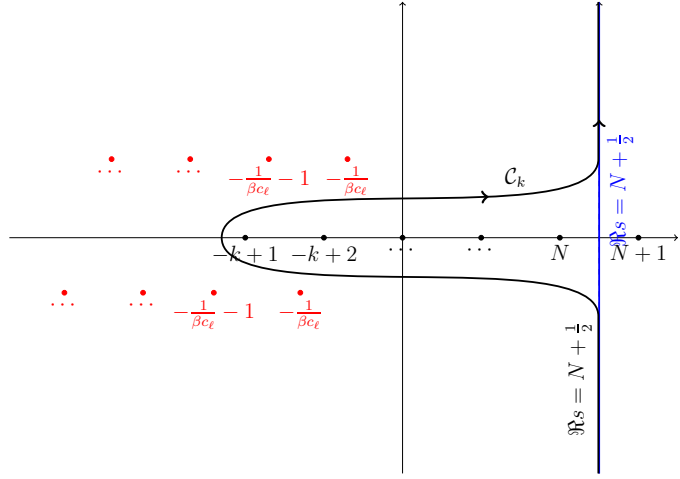
\begin{figure}[H]
\begin{center}
\resizebox{0.6\textwidth}{!}{
\begin{tikzpicture}[scale=01.5]

\draw[fill] (0,0) node[below]{$-k+1$} circle [radius=0.03];
\draw[fill] (1,0) node[below]{$-k+2$} circle [radius=0.03];
\draw[fill] (2,0) node[below]{$\cdots$} circle [radius=0.03];
\draw[fill] (3,0) node[below]{$\cdots $} circle [radius=0.03];
\draw[fill] (4,0) node[below]{$N$} circle [radius=0.03];
\draw[fill] (5,0) node[below]{$N+1$} circle [radius=0.03];

\draw[fill,red] (1.3,1) node[below]{$-\frac 1{\beta c_\ell}$} circle [radius=0.03];
\draw[fill,red] (0.3,1) node[below]{$-\frac 1{\beta c_\ell}-1$} circle [radius=0.03];
\draw[fill,red] (-0.7,1) node[below]{$\cdots$} circle [radius=0.03];
\draw[fill,red] (-1.7,1) node[below]{$\cdots $} circle [radius=0.03];

\draw[fill,red] (0.7,-0.7) node[below]{$-\frac 1{\beta c_\ell}$} circle [radius=0.03];
\draw[fill,red] (-0.4,-0.7) node[below]{$-\frac 1{\beta c_\ell}-1$} circle [radius=0.03];
\draw[fill,red] (-1.3,-0.7) node[below]{$\cdots$} circle [radius=0.03];
\draw[fill,red] (-2.3,-0.7) node[below]{$\cdots $} circle [radius=0.03];

\draw[->] (2,-3)--(2,3);
\draw[->] (-3,0)--(5.5,0);
\draw[dashed, -> ] (4.5,-3)--node[pos =0.2,sloped, above] {$\Re s= N+\frac 1 2$} (4.5,3);

%

\draw[ blue,  line width=1,postaction={decorate,decoration={{markings,mark=at position 0.75 with {\arrow[black,line width=1.5pt]{>}}}} }]
(4.5,-3)--node[pos=0.6, below, sloped]{$\Re s = N+\frac 1 2 $}(4.5, 3);

\draw[black , line width =1,postaction={decorate,decoration={{markings,mark=at position 0.75 with {\arrow[black,line width=1.5pt]{>}}}} }]
(4.49,-3) to (4.49, -1) to[out=90, in = -90, looseness=0.5] (-0.3,0)
to [out=90, in =-90, looseness=0.5] node [pos=0.7, sloped, above]{ $\mathcal C_k$} (4.49, 1) to (4.49,3);
;

%
\end{tikzpicture}}
\end{center}
\caption{The contours of integration for the function \eqref{phikj} in the case $L>M+1$.}
\label{Ck}
\end{figure}

Next, regardless of the relative valus of $L$ and $M$, applying the operator $\DD +k -1$ to 
$\tilde{\phi}^{({\bf c},{\bf d},\beta)}_k(x)$  and evaluating
under the integral sign  it follows immediately that the sequence of equations (\ref{D_phi_k}) is satisfied. 

Finally, the following computation shows that the sequence of eigenfunction equations (\ref{phi_k_eq})
 is satisfied by the $\tilde{\phi}^{({\bf c},{\bf d},\beta)}_k(x)$ 's, and hence, so are  the recursion relations (\ref{phi_k_rec}).
 \bea
&\& ( \DD + k -1) \int_{\CC_k}  \frac {\Gamma(1-k-s)  \prod_{\ell=1}^L \Gamma \left( s + \frac 1 {\beta c_\ell} \right) 
\left(-\kappa_{{\bf c}, {\bf d}} x\right)^s}{\prod_{m=1}^M \Gamma\left( s - \frac 1{\beta d_m}  \right)}  ds. \cr
 &\& = \int_{\CC_k}  \frac { (s + k -1)\Gamma(1-k-s)  \prod_{\ell=1}^L \Gamma \left( s + \frac 1 {\beta c_\ell} \right) 
\left(-\kappa_{{\bf c}, {\bf d}} x\right)^s}{\prod_{m=1}^M \Gamma\left( s - \frac 1{\beta d_m}  \right)}  ds \cr
&\&=- \int_{\CC_k}  \frac {\Gamma(2-k-s)  \prod_{\ell=1}^L \Gamma \left( s + \frac 1 {\beta c_\ell} \right) 
\left(-\kappa_{{\bf c}, {\bf d}} x\right)^s}{\prod_{m=1}^M \Gamma\left( s - \frac 1{\beta d_m}  \right)}  ds 
\cr
&\&=- \int_{\CC_k}  \frac {\Gamma(1-k-s)  \prod_{\ell=1}^L \Gamma \left(1+ s + \frac 1 {\beta c_\ell} \right) 
\left(-\kappa_{{\bf c}, {\bf d}} x\right)^{s+1}}{\prod_{m=1}^M \Gamma\left( 1+ s - \frac 1{\beta d_m}  \right)}  ds \cr
&\&= 
\kappa_{{\bf c}, {\bf d}} x \int_{\CC_k}  \frac {\Gamma(1-k-s)  \prod_{\ell=1}^L \Gamma \left(1+ s + \frac 1 {\beta c_\ell} \right) 
\left(-\kappa_{{\bf c}, {\bf d}} x\right)^s}{\prod_{m=1}^M \Gamma\left( 1+ s - \frac 1{\beta d_m}  \right)}  ds \cr
&\&= 
\kappa_{{\bf c}, {\bf d}} x \int_{\CC_k}  \frac {\Gamma(1-k-s)  \prod_{\ell=1}^L \left( s + \frac 1 {\beta c_\ell} \right)\Gamma \left( s + \frac 1 {\beta c_\ell} \right) 
\left(-\kappa_{{\bf c}, {\bf d}} x\right)^s}{\prod_{m=1}^M \left( s - \frac 1{\beta d_m} \right)\Gamma\left( s - \frac 1{\beta d_m}  \right)}  ds \cr
&\&= 
x \int_{\CC_k}  \frac {\Gamma(1-k-s)  \prod_{\ell=1}^L \left( 1+  \beta  s c_\ell  \right)\Gamma \left( s + \frac 1 {\beta c_\ell} \right) 
\left(-\kappa_{{\bf c}, {\bf d}} x\right)^s}{\prod_{m=1}^M \left( 1 -\beta s d_m \right)\Gamma\left( s - \frac 1{\beta d_m}  \right)}  ds \cr
&\& = x G_{{\bf c}, {\bf d}}(\beta \DD)  \int_{\CC_k}  \frac {\Gamma(1-k-s)  \prod_{\ell=1}^L \Gamma \left( s + \frac 1 {\beta c_\ell} \right) 
\left(-\kappa_{{\bf c}, {\bf d}} x\right)^s}{\prod_{m=1}^M \Gamma\left( s - \frac 1{\beta d_m}  \right)}  ds,
 \eea
 where in the fourth line, the contour $\CC_k$ is understood to be chosen such that the shift $s \ra s+1$ does not traverse any poles.
\end{proof}

\begin{remark}
\label{slater_L_smaller} 
For $L \le M+1$, we could equally well have  replaced 
\be
G_{L,M+1}^{1,L}\left( {1-\frac 1{\beta c_1}, \cdots, 1-\frac 1{\beta c_L} \atop 1 -k , 1+\frac 1 {\beta d_1}, \cdots, 1+\frac 1{\beta d_M}}  \bigg{|}  -\kappa_{{\bf c}, {\bf d}} x \right)
\ee  
in  (\ref{phi_tilde_k_integ}) by 
\be 
G^{1,n}_{L,M+1}\left( {1-\frac 1{\beta c_1}, \cdots, 1-\frac 1{\beta c_L} \atop  1 -k ,
1+\frac 1 {\beta d_1}, \cdots, 1+\frac 1{\beta d_M}}  \bigg{|} (-1)^{L+n+1} \kappa_{{\bf c}, {\bf d}} x \right)
\ee
for any $0 \le n\le L$, with a suitably modified normalization constant, since these  are all constant multiples of each other, as can be seen, e.g., from Slater's identity \cite{NIST}:
\bea 
&\&G^{1,n}_{L,M+1}\left( {1-\frac 1{\beta c_1}, \cdots, 1-\frac 1{\beta c_L} \atop  1 -k ,
1+\frac 1 {\beta d_1}, \cdots, 1+\frac 1{\beta d_M}}  \Big{|}\zeta \right) \cr
&\& =   { \zeta^{1-k} \prod_{l=1}^n \Gamma({\script 1-k +{1\over \beta c_l}})
\ \leftidx{_L}{F_M}\left( {1-k +\frac 1{\beta c_1}, \cdots, 1-k +\frac 1{\beta c_L} \atop  
1- k-\frac 1 {\beta d_1}, \cdots, 1-k -\frac 1{\beta d_M}}  \Big{|} (-1)^{L+n+1} \zeta \right) \over
\prod_{l=n+1}^L\Gamma({\script k - {1\over \beta c_l}}) \prod_{m=1}^M\Gamma({\script 1-k - {1\over \beta d_m}}) },
\cr 
&\&
\label{slater_id_L_smaller}
\eea
The equality (\ref{phi_k_tilde_phi_k}) also follows  in this case from (\ref{slater_id_L_smaller}) with $n=L$.

 For $n=L$, $k=1$, this reproduces  (\ref{sum_Gamma_residues_0}). For general $k$, it
 reproduces (\ref{phi_k_tilde_phi_k}), with $\phi_k^{({\bf c}, {\bf d}, \beta)}(x)$ 
 defined by the series (\ref{phi_k_series_ratl}) which, up to normalization, is just $x^{1-k}$ times the 
 hypergeometric series  for
 \be
 \leftidx{_L}{F_M}\left({1-k +\frac 1{\beta c_1}, \cdots, 1-k +\frac 1{\beta c_L} \atop  
1- k-\frac 1 {\beta d_1}, \cdots, 1-k -\frac 1{\beta d_M}}  \Big{|} (-1)^{L+n} \kappa_{{\bf c}, {\bf d}} x \right).
  \ee
  \end{remark}
  
  \begin{remark}
  \label{slater_L_larger}
 
If $L>M+1$, we could also have replaced 
\be
G_{L,M+1}^{1,L}\left( {1-\frac 1{\beta c_1}, \cdots, 1-\frac 1{\beta c_L} \atop 1 -k,
1+\frac 1 {\beta d_1}, \cdots, 1+\frac 1{\beta d_M}}  \bigg{|} - \kappa_{{\bf c}, {\bf d}} x \right)
\ee 
in  (\ref{phi_tilde_k_integ}) by any suitably normalized  linear combination of the linearly independent functions
\be
G^{1,{L-1}}_{L,M+1}\left({1-{1\over \beta c_1}, \cdots,  \widehat{1-{1\over \beta c_i}}, 
\cdots , 1-{1\over \beta c_L} ,  1-{1\over \beta c_i }
\atop 1-k, 1+{1\over \beta d_1},   \cdots , 1+{1\over \beta d_M} } 
 \Big | \, \kappa_{{\bf c}, {\bf d}} x\right),
 \quad i=1, \dots L,
 \label{slater_id_L_larger}
\ee
which form a basis of solutions of the same differential equation and, within a normalization, 
have the same asymptotic series (\ref{phi_k_series_ratl}) within the angular sector
\be
0< \arg(\kappa_{{\bf c}, {\bf d}} x) <  (L-M-1){\pi \over 2}
\ee
\end{remark}

To see the relation between  eq.~(\ref{phi_k_eq}) and the equations satisfied by the Meijer $G$-functions, we set
\be
\zeta:=- \kappa_{{\bf c}, {\bf d}} x
\ee
and rewrite (\ref{phi_k_eq})  as
\be
 \zeta {\prod_{l=1}^L (\DD + {1\over \beta c_l })\over \prod_{m=1}^M(\DD -{1\over \beta d_m} )}\phi^{({\bf c},{\bf d},\beta)}_k
  + (\DD+k-1) \phi^{({\bf c},{\bf d},\beta)}_k =0,
\label{phi_k_eq_inv_zeta}
\ee
where
\be
\DD = \zeta{d\over d\zeta} = x {d\over dx}.
\ee
Multiplying (\ref{phi_k_eq_inv_zeta})  by $\zeta^{-1}$ and applying the 
operator $\prod_{m=1}^M(\DD -{1\over \beta d_m})$ on the left gives
\be
 \prod_{l=1}^L (\DD + {1\over \beta c_l })\phi^{({\bf c},{\bf d},\beta)}_k
  + \prod_{m=1}^M(\DD -{1\over \beta d_m} )\zeta^{-1}(\DD+k-1) \phi^{({\bf c},{\bf d},\beta)}_k =0.
\label{phi_k_eq_dir_zeta}
\ee
Now using the ladder  relation
\be
\left(\DD- {1\over \beta d_j} \right) \zeta^{-1} = \zeta^{-1}\left (\DD- {1\over \beta d_j} +1\right) 
\ee
repeatedly to move  $\zeta^{-1}$ in the second term in (\ref{phi_k_eq_dir_zeta})
 to the left of the differential operator term and, finally,  multiplying by $\zeta$ gives
\be
  \zeta \prod_{l=1}^L (\DD + {1\over \beta c_l })\phi^{({\bf c},{\bf d},\beta)}_k
  + (\DD+k-1)  \prod_{m=1}^M(\DD  -1 -{1\over \beta d_m} )\phi^{({\bf c},{\bf d},\beta)}_k =0,
\label{phi_k_eq_dir_zeta_meijer}
\ee
which is eq.~(\ref{MeijerGeq}) for
\bea
p&\&=L, \quad q= M+1, \quad m=1, \quad n=L \cr
&\& \cr
a_l &\&= 1 - {1\over \beta c_l},\ l=1, \dots, L, \cr
b_1&\& = 1-k, \ b_{m+1} = 1 + {1\over \beta d_m}, \ m=1, \dots, M.
\eea

In particular, for the case  $L=2$, $M=0$, with  three branch points, two of them, with branching profiles $(\mu^{(1)}, \mu^{(2)})$,
 weighted as in  eq.~(\ref{WG_def}), and a third unweighted one with profile $\mu$, we have
\be
\tilde{\phi}^{((c_1, c_2), {\bf \cdot}, \beta)}_k(x) = C_k^{((c_1, c_2), {\bf \cdot}, \beta)}
G^{1,2}_{2,1}\left( {\script{1-{1\over \beta c_1}\  1-{1\over \beta c_2} 
\atop 1-k  }} \Big | \kappa_{{\bf c}, {\bf d}} x\right) .
\label{belyi_G_function}
\ee
This corresponds to the enumeration of {\em Belyi curves} or, equivalently,  {\em dessins d'enfants} \cite{AC1, Z}.


\section{The $\tau$-function $\tau^{(G,\beta)}({\bf t})$ evaluated at power sums}
\label{tau_trace_invars}

As detailed in \cite{ACEH1, ACEH2, ACEH3},  $\tau^{(G,\beta)}({\bf t})$ is the KP $\tau$-function  corresponding to 
the Grassmannian element  spanned by the basis elements $\{\phi_k(1/z)\}_{k \in \Nb^+}$
\be
w^{(G, \beta)}  = \Span\{\phi_k(1/z) \},  \ k\geq 1.
\ee

It follows  \cite{HB} that, if $\tau^{(G,\beta)}({\bf t})$  is evaluated at the trace invariants
\be
{\bf t} = \big[ X \big], \quad t_i = {1\over i} \tr X^i
\ee
of a diagonal $n \times n$ matrix
\be
X := \diag(x_1, \dots, x_n),
\label{X_def}
\ee
it is expressible as the ratio of $n \times n$ determinants 
\be
\tau^{(G,\beta)}\left(\big[X \big]\right) ={\prod_{i=1}^n x_i^{n-1}\over \prod_{i=1}^n \rho_{-i}} {\det\left( \phi_i(x_j)\right)_{1\leq i,j, \leq n} \over \Delta(x)},
\label{tau_phi_i_det}
\ee
where
\be
\Delta(x) = \prod_{1\leq i < j \leq n}(x_i - x_j) = \det(x_i^{n-j})_{1\leq i,j, \leq n} 
\ee
is the Vandermonde determinant.
\begin{remark}
The evaluation $\tau^{(G,\beta)}\left(\big[X \big]\right)$ is precisely what was originally defined in \cite{GR} as 
a ``hypergeometric function of matrix argument''. See also \cite{Or} for matrix integral representations of $\tau$-functions.
\end{remark}

It also follows from (\ref{D_phi_k}) that each $\{\phi^{({\bf c},{\bf d},\beta)}_k(x)\}_{1\leq k \leq n}$ may be expressed as a finite lower triangular linear combination of the powers 
of the Euler operator $\DD$ applied to $\phi^{({\bf c},{\bf d},\beta)}_n(x)$
\be
\phi_k(x) = \beta^{n-k}\DD^{n-k} \phi_n(x) + \sum_{j=0}^{n-k-1} \Gamma_{k j} \DD^j \phi_n(x), \quad  k=1, \dots, n.
\ee
where the $\Gamma_{kj}$'s are polynomials in $\beta$ of lower degree  with integer coefficients.
Therefore
\be
\tau^{(G,\beta)}\left(\big[X \big]\right) =\gamma_n \left(\prod_{i=1}^n x_i^{n-1}\right){\det\left( \DD^{i-1}\phi_n(x_j)\right)_{1\leq i,j, \leq n} \over \Delta(x)},
\label{tau_phi_i_det}
\ee
where
\be
\gamma_n :={\beta^{{1\over 2}n(n-1)}\over \prod_{i=1}^n \rho_{-i}}.
\label{kappa_n}
\ee

For  rational weight generating functions $G=G_{{\bf c}, {\bf d}}$,  and any positive integer $n\in \Nb^+$, let
\be
f_{({\bf c}, {\bf d}, \beta, n)}(y) := \tilde{\phi}^{({\bf c}, {\bf d}, \beta)}_n(e^y) = \int_{\CC_n} A^{({\bf c}, {\bf d}, \beta)}_n(s) e^{ys} ds,
\label{f_c_d_beta_n}
\ee
\be
A^{({\bf c}, {\bf d}, \beta)}_n(s) := 
{ C_n^{({\bf c},{\bf d},\beta)}\Gamma(1-n-s) \prod_{l=1}^L \Gamma \left(s + {1 \over \beta c_l }\right) (-\kappa_{{\bf c}, {\bf d}})^s \over 
 2\pi i  \prod_{m=1}^M \Gamma\left( s- \frac 1{\beta d_m} \right) }.
\label{A_n_s}
\ee

Defining the diagonal matrix 
\be
Y = \diag (y_1, \dots y_n)
\label{Y_def}
\ee
 by
\be
X = e^Y, \quad Y=\ln(X), \quad x_i = e^{y_i}, \quad i =1 , \dots, n,
\ee
eq.~(\ref{tau_phi_i_det}) can be expressed as a ratio of Wronskian determinants
\be
\tau^{(G_{{\bf c}, {\bf d}},\beta)}\left(\big[X \big]\right) =\gamma_n  \left(\prod_{i=1}^n x_i^{n-1}\right){\det\left( f_{({\bf c}, {\bf d}, \beta, n)}^{(i-1)}(y_j)\right)_{1\leq i,j, \leq n} \over \Delta(e^y)},
\label{tau_wronskian_f_det}
\ee
where
\be
\Delta(e^y):= \prod_{1\leq i< j \leq n}(e^{y_i} - e^{y_j}).
\ee


\section{Matrix integral representation of $\tau^{(G_{{\bf c, \bf d}},\beta)}\left(\big[X \big]\right)$}
\label{MM_tau}

Consider the generalized Br\'ezin-Hikami \cite{BH} matrix integral
\be
\Zb_{d\mu_{({\bf c}, {\bf d}, \beta, n)}}(X) =\int_{M \in \Nor^{n\times n}_{\CC_n}}d\mu_{({\bf c}, {\bf d}, \beta, n)}(M) e^{\tr Y M},
\label{brezin_Hikami_Y}
\ee
where
\be
d\mu_{({\bf c}, {\bf d}, \beta, n)} (M):= (\Delta({\bf \zeta})^2 \det(A^{({\bf c}, {\bf d}, \beta)}_n(M)) d\mu_0(U)\prod_{j=1}^n d\zeta_i  
\ee
is a conjugation invariant measure on the space $\Nor^{n\times n}_{\CC_n}$ of $n \times n$ normal matrices 
\be
M = U Z U^\dag \in \Nor^{n\times n}_{\CC_n}, \quad U \in U(n), \quad Z= \diag(\zeta_1, \dots, \zeta_n)
\ee
with eigenvalues $\zeta_i \in \Cb$ supported on the contour $\CC_n$  and $d\mu_0(U)$ is the Haar measure on $U(n)$.
\begin{theorem}
\label{hypergeom_tau_matrix_integral}
The $\tau$-function $\tau^{(G_{{\bf c}, {\bf d}},\beta)} ({\bf t})$, evaluated at the trace invariants $\big[ X \big]$
of the externally coupled matrix $X$ is given by
\be
\tau^{(G_{{\bf c}, {\bf d}},\beta)} (\big[ X \big])= {\beta^{{1\over 2}n(n-1)} (\prod_{i=1}^n x_i^{n-1})\Delta(\ln(x)) \over (\prod_{i=1}^n i!) 
\Delta(x) } \Zb_{d\mu_{({\bf c}, {\bf d}, \beta, n)} }(X),
\label{tau_Gcb_Z_dmu}
\ee
where $\Delta(\ln(x))= \Delta(y))$ is the Vandermonde determinant in the variables $(y_1, \dots, y_n)$. 
\end{theorem}
\begin{proof}
Using the Harish-Chandra-Itzykson-Zuber identity  \cite{HC, IZ} 
\be
\int_{U \in U(n)}d\mu(U) e^{\tr(YUZ U^\dag)} = {(\prod_{i=1}^{n-1} i!) \det\left(e^{y_i \zeta_j}\right) \over \Delta(y) \Delta(\zeta)}
\label{HCIZ_integral}
\ee
to evaluate the angular integral gives
\bea
\Zb_{d\mu_{({\bf c}, {\bf d}, \beta, n)}}(X) &\&= {(\prod_{i=1}^{n-1} i!)\over \Delta(y)}\prod_{i=1}^n\left(\int_{\CC_n}  d\zeta_i A^{({\bf c}, {\bf d}, \beta)}_n(\zeta_i) \right)\Delta(\zeta) \det\left( e^{y_i \zeta_j}\right)_{1\leq i, j, \leq n} \cr
&\& = { (\prod_{i=1}^n i!)\over \Delta(y)}\det\left(f_{({\bf c}, {\bf d}, \beta, n)}^{(i-1)}(y_j)\right)_{1\leq i,j, \leq n},
\label{det_rep_Z_d_mu}
\eea
where we have used the Andr\'eiev identity \cite{An} in the second line.
Comparing eqs.~(\ref{det_rep_Z_d_mu}) and (\ref{tau_wronskian_f_det}), we obtain the matrix integral 
representation (\ref{tau_Gcb_Z_dmu})   of $\tau^{(G_{{\bf c}, {\bf d}},\beta)}(\big[X\big])$.
\end{proof}

\break
 \bigskip
\noindent 
\small{ {\it Acknowledgements.} 
This work was partially supported by the Natural Sciences and Engineering Research Council of Canada (NSERC) 
and the Fonds de recherche du Qu\'ebec, Nature et technologies (FRQNT).  
\bigskip


\newcommand{\arxiv}[1]{\href{http://arxiv.org/abs/#1}{arXiv:{#1}}}

\bigskip
\noindent

\end{document}